\documentclass[pra,twocolumn,superscriptaddress,nofootinbib,longbibliography]{revtex4-1}
\usepackage[bookmarks = false, pdfpagemode = None, pdfstartview = FitH, colorlinks = true, urlcolor = blue,hyperfootnotes=false]{hyperref}
\usepackage{amsmath, amsthm, amssymb}
\usepackage{graphicx}
\usepackage{xypic}
\usepackage{rotating}
\usepackage{color}
\usepackage{pgf}
\usepackage{comment}
\input{Qcircuit}

\newcommand{\code}[1]{\mathcal{#1}}
\newcommand{\C}[1]{{{\vphantom{#1}}^{C}\hspace{-.3em}{#1}}}
\newcommand{\CX}{\C{X}}
\newcommand{\CY}{{{\vphantom{Y}}^{C}\hspace{-.1em}{Y}}}
\newcommand{\CZ}{{{\vphantom{Z}}^{C}\hspace{-.2em}{Z}}}
\newcommand{\CH}{\C{H}}
\newcommand{\opfunc}[1]{{\mathcal{#1}}}
\newcommand{\classicallogicgate}[1]{{\texttt{#1}}}
\definecolor{LightCyan}{rgb}{0.45,1,1}

\newtheorem{theorem}{Theorem}[section]
\newtheorem{lemma}[theorem]{Lemma}

\newtheorem{corollary}[theorem]{Corollary}

\renewcommand{\cD}{{\code{D}}}

\newcommand{\SWAP}{{
  \begin{pgfpicture}{0cm}{0cm}{.265cm}{.226cm}
  \pgfsetlinewidth{.8pt}
  \pgfxyline(0,.226)(.1,.226)
  \pgfxyline(0,.0145)(.1,.0145)
  \pgfxyline(.165,.226)(.265,.226)
  \pgfxyline(.165,.0145)(.265,.0145)
  \pgfxyline(.0875,.235)(.1775,.0055)
  \pgfxyline(.0875,.0055)(.1775,.235)
  \end{pgfpicture}
  }}

\begin{document}

\title{Magic-state distillation with the four-qubit code}

\author{Adam M. Meier}
\email[]{Adam.Meier@colorado.edu}
\affiliation{University of Colorado, Boulder, CO}
\affiliation{National Institute of Standards and Technology, Boulder, CO}
\author{Bryan Eastin}
\affiliation{Northrop Grumman Corporation, Baltimore, MD}
\author{Emanuel Knill}
\affiliation{University of Colorado, Boulder, CO}
\affiliation{National Institute of Standards and Technology, Boulder, CO}

\date{\today}

\begin{abstract}
The distillation of magic states is an often-cited technique for
enabling universal quantum computing once the error probability for a
special subset of gates has been made negligible by other means. We
present a routine for magic-state distillation that reduces the
required overhead for a range of parameters of practical
interest. Each iteration of the routine uses a four-qubit
error-detecting code to distill the $+1$ eigenstate of the Hadamard
gate at a cost of ten input states per two improved output
states. Use of this routine in combination with the $15$-to-$1$ distillation
routine described by Bravyi and Kitaev allows for further improvements
in overhead.
\end{abstract}


\maketitle

Many techniques for robustly implementing quantum gates most naturally
generate only a finite subset of the unitary operators.  Frequently, the
naturally convenient quantum operations generate the full set of
Clifford operations, which consists of the Clifford group of unitaries
augmented by measurement and state preparation in the standard
basis. Clifford operations are sufficient for stabilizer-state
preparations and measurements and thus underlie stabilizer-based error
correction and much of the associated theory of fault tolerance.
Though inadequate for universal quantum computing, the Clifford
operations can be supplemented by any unitary outside of the Clifford 
group to obtain a
universal set~\cite{Nebe2001}.  Consequently, the problem of achieving
universality is often reduced to that of finding a way of robustly
implementing a single non-Clifford unitary gate. 

Given the ability to perform Clifford operations, non-Clifford gates
can be indirectly implemented using certain non-stabilizer states as a
consumable resource.  The advantage of this approach lies in the
possibility of distilling such resource states prior to use.
Distillation is a technique whereby a collection of
independently prepared faulty resource states can be converted into a
smaller number of resource states whose fidelity with respect to the
ideal state is higher.  Some states have the property that one can
distill them using only Clifford operations. States that are both
sufficient for universality and distillable in this way are known as
\textit{magic states}.  Magic-state distillation allows faulty magic
states to be used as a resource for robust universal quantum
computing.

The notion of magic states was introduced by Bravyi and
Kitaev~\cite{BravyiKitaev2005}, who showed that the (magic) 
eigenstates of the
one-qubit Clifford gates $T$ and $H$ can be distilled from copies of
these states with error probabilities of up to $0.173$ and
$0.141$ per state, respectively.  Their distillation routines work
by projecting several such faulty copies of a specified magic state
(henceforth, resource state) into a stabilizer code and then decoding
the result, checking for and discarding on any indication of error.
Distillation of the $T$-eigenstate $\ket{T}$ employs a projection onto
the $5$-qubit distance-$3$ code, while distillation of the $H$-eigenstate $\ket{H}$ relies on the $15$-qubit Reed-Muller code;
both distillation routines result in one improved resource state.  We
refer to the $\ket{H}$-distillation routine as the $15$-to-$1$
routine.

An apparently distinct routine for distilling $\ket{H}$ using the
$7$-qubit Steane code was proposed previously by
Knill~\cite{Knill2004}, but Reichardt found the two routines to be
equivalent~\cite{Reichardt2005}.  Reichardt additionally showed that 
the error threshold for distilling $\ket{H}$ could be improved from
$0.141$ to $0.146$ via a $7$-to-$1$ distillation routine, thereby proving that every faulty $\ket{H}$
outside of the set of stabilizer states is distillable with a finite routine. 
Campbell and Browne proved the
impossibility of a similar result for $\ket{T}$ by showing that no finite distillation
routine is capable of distilling faulty $\ket{T}$ arbitrarily near the boundary of
the stabilizer states~\cite{CampbellBrowne2009,CampbellBrowne2010}.

The focus of each of the aforementioned papers is on the threshold for
magic-state distillation, but the efficiency of a distillation routine
is crucial to its practical utility.  Of particular concern is the
number of faulty resource states required as input to distill each
resource state of some desired quality.  This ratio contributes
strongly to the overhead required to implement a quantum computation
using magic-state distillation~\cite{Jones2010}, potentially
increasing the number of qubits and gates required by a large
multiplicative factor.  With this in mind, we describe a routine for
distilling $\ket{H}$ that reduces the number of input resource states
required per output state, distilling $2$ improved resource states
from $10$.  The routine can be used either solely or in combination
with previously developed routines to obtain resource reductions for a
variety of parameter ranges of interest.

After explaining the needed background in Sec.~\ref{sec:background},
we introduce and analyze the proposed distillation routine in 
Secs.~\ref{sec:HState} and \ref{sec:analysis} and compare it to the
$15$-to-$1$ routine in Sec.~\ref{sec:comparativePerformance}.
Sec.~\ref{sec:distillationSequences} explains how sequential distillation rounds can be combined.  Concluding remarks appear in
Sec.~\ref{sec:conclusions}.

\section{Basic Concepts and Notation\label{sec:background}}

We denote the one-qubit Pauli operators by $I$, $X$, $Y$, and $Z$, where $I$ is
the identity and the others are $i$ times the conventional
$\pi$-rotations associated with the eponymous axes of the Bloch
sphere.  The $n$-qubit Pauli group consists of all $n$-fold tensor
products of the one-qubit Pauli operators multiplied by $\{\pm 1, \pm i\}$.

The Clifford group consists of the unitary operators that normalize
the Pauli group.  Any Clifford unitary can be constructed 
by composition of tensor products of $H$ (Hadamard), $T$, and $\CX$
(controlled-\classicallogicgate{NOT} or controlled-$X$) gates up
to an unimportant global phase.  In the
standard basis, these gates are given by
\begin{align*}
H = \frac{1}{\sqrt{2}}\left( \begin{array}{cc}
	1 & 1 \\
	1 & -1
\end{array} \right ) \;,&&
T = \frac{e^{i \pi/4}}{\sqrt{2}}\left( \begin{array}{cc}
	1 & 1 \\
	i & -i
\end{array} \right ) \; , 
\end{align*}
and
\begin{equation*}
\CX = \left( \begin{array}{cccc}
	1 & 0 & 0 & 0 \\
	0 & 1 & 0& 0 \\
	0 & 0 & 0& 1 \\
	0 & 0 & 1 & 0 \\
\end{array} \right ) \;.
\end{equation*}
For convenience, we additionally employ the following one- and
two-qubit Clifford gates: $P(\frac{\pm\pi}{2})=e^{\mp i P\pi/4}$ where
$P\in\{X,Y,Z\}$ ($\pi/2$ rotations about the axes of the Bloch
sphere), $S=e^{i \pi/4} Z(\frac{\pi}{2})$, $\CZ$ (controlled-$Z$),
$\CY$ (controlled-$Y$), and $\SWAP$ (\classicallogicgate{SWAP}).  We
also make frequent use of three unitaries not contained in the Clifford group: $Y(\frac{\pm
  \pi}{4})$ (the $\pm \pi/4$ rotations about the $Y$ axis) and $\CH$
(controlled-$H$).

We use the term ``Clifford operation'' to refer to any quantum 
operation that can be implemented using Clifford unitaries together
with preparation and measurement in the standard basis.  States that
can be prepared with Clifford operations are known as stabilizer
states.  Pure stabilizer states are $+1$ eigenstates of a complete set
of commuting (generally multi-qubit) Pauli operators; this set of
Pauli operators is known as the stabilizer generator.  Similarly, one
can define a quantum (stabilizer) code as the $+1$ eigenspace of a
non-maximal set of stabilizers.

Whenever necessary, subscripts on operators are used to identify the
qubits that they act upon.  A bar over an operator indicates that it
is a logical (or encoded) operator. That is, it acts as the specified
operator on qubits encoded in a quantum code.

Quantum circuit diagrams in this paper conform to the notation used in
Ref.~\cite{Nielsen2001} with the following exceptions: Wires
representing multiple qubits are not specially decorated, and the
symbols
\begin{align*}
\raisebox{.75em}{
\Qcircuit @C=1em @R=1em {
& \ctrl{1} & \qw \\
& \control \qw & \qw
}
}\;,
&&
\Qcircuit @C=1em @R=1em {
& \measureD{X}
}
\;,
&&
\text{and}
&&
\Qcircuit @C=1em @R=1em {
& \measureD{Z}
}
\end{align*}
are used to represent the $\CZ$ gate and projective measurements in
the eigenbases of $X$ and $Z$, respectively.

\begin{figure}
\begin{tabular}{ll}
(a) &
\raisebox{0em}{ 
\Qcircuit @C=.5em @R=1em {
& & & & \lstick{\ket{+}} & \ctrl{1}  & \measureD{X} \\
& \qw & \qw & \qw & \qw & \gate{H} & \qw & \qw
}
} \\ \\
(b) &
    \raisebox{0em}{
    \Qcircuit @C=.5em @R=1em {
      & \ctrl{1} & \qw \\
      \push{\rule{0em}{1.6em}} & \gate{H} & \qw
    }
    }
    \raisebox{-1.2em}{\rule{1em}{0em}$=$\rule{.8em}{0em}}
    \raisebox{0em}{
    \Qcircuit @C=.5em @R=1em {
      & \qw & \ctrl{1} & \qw & \qw \\
      & \gate{Y(\frac{-\pi}{4})} & \control \qw & \gate{Y(\frac{\pi}{4})} & \qw
    }
    }
    \\ \\
(c) &
\raisebox{-1.3em}{
\Qcircuit @C=.5em @R=1em {
& \gate{Y(\frac{\pm \pi}{4})} & \qw
} }
\raisebox{-1.5em}{\rule{1em}{0em}$=$\rule{.8em}{0em}}
\raisebox{0em}{
\Qcircuit @C=.5em @R=1em {
& & & &\lstick{\ket{H}} & \ctrl{1} & \measureD{\pm Y} \cwx[1] \\
& \qw & \qw & \qw & \qw & \gate{Y} & \gate{ Y(\frac{\pm\pi}{2})} & \qw
}
}
\end{tabular}
 \caption{Circuits showing that the Hadamard operator can be measured
   non-destructively using only Clifford operations and two copies of
   $\ket{H}$.  The circuit in (a) implements a non-destructive
   measurement of the Hadamard operator using the non-Clifford $\protect\CH$ gate, while (b) and (c) give circuit
   identities that can be used to break the measurement up into
   Clifford operations and $\ket{H}$ resource states.  The classical control
   in (c) is meant to indicate (for the $+$ case) that a positive Y measurement
   triggers a $Y(\frac{\pi}{2})$ gate while a negative measurement triggers the 
   identity. \label{fig:littleCircuits}
 }
\end{figure}
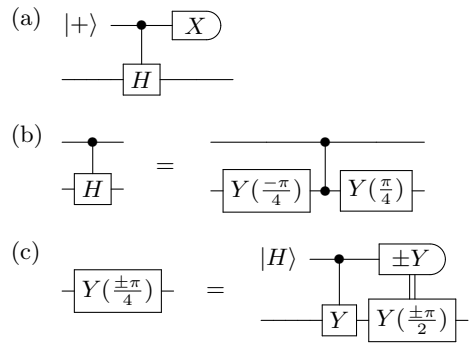

\begin{figure*}
\begin{tabular}{ll}
(a) & 
\raisebox{0em}{
\Qcircuit @C=.7em @R=.7em {
& & & \lstick{\ket{+}} & \ctrl{1} & \ctrl{2} & \measureD{X} \\
& \qw & \qw & \qw & \gate{H} & \qw & \qw & \qw \\
& \qw & \qw & \qw & \qw & \gate{H} & \gate{H} & \qw
}}
\raisebox{-2.2em}{\rule{1.6em}{0em}$=$\rule{1.6em}{0em}} 
\begin{pgfpicture}{0cm}{0cm}{0cm}{0cm}
\color{green}
\pgfrect[fill]{\pgfpoint{.9cm}{.2cm}}{\pgfpoint{1.75cm}{-1.75cm}}
\pgfrect[fill]{\pgfpoint{3.37cm}{.2cm}}{\pgfpoint{1.75cm}{-1.75cm}}
\color{LightCyan}
\pgfrect[fill]{\pgfpoint{2.65cm}{.2cm}}{\pgfpoint{.72cm}{-1.75cm}}
\end{pgfpicture}
\Qcircuit @C=.7em @R=.7em {
& & & \lstick{\ket{+}} & \ctrl{1} & \ctrl{2} & \ctrl{1} & \qw & \ctrl{1} & \ctrl{2} & \ctrl{1} & \measureD{X} \\
& \qw & \qw & \qw & \gate{H} & \qw & \qswap \qwx[1] & \qw  & \gate{H} & \qw & \qswap \qwx[1] & \qw & \qw \\
& \qw & \qw & \qw & \qw & \gate{H} & \qswap & \gate{H} & \qw & \gate{H} & \qswap & \qw & \qw \\
} \\
\\
(b) &
\begin{pgfpicture}{0cm}{0cm}{0cm}{0cm}
\color{green}
\pgfrect[fill]{\pgfpoint{.9cm}{.2cm}}{\pgfpoint{2.87cm}{-3.1cm}}
\pgfrect[fill]{\pgfpoint{6.97cm}{.2cm}}{\pgfpoint{2.87cm}{-3.1cm}}
\color{LightCyan}
\pgfrect[fill]{\pgfpoint{3.77cm}{.2cm}}{\pgfpoint{3.2cm}{-3.1cm}}
\end{pgfpicture}
\Qcircuit @C=.7em @R=.7em {
& & & \lstick{\ket{+}} & \ctrl{1} & \ctrl{2} & \ctrl{3} & \ctrl{4} & \qw & \qw & \qw & \qw & \qw & \ctrl{1} & \ctrl{2} & \ctrl{3} & \ctrl{4} & \measureD{X} \\
& \qw & \qw & \qw & \gate{H} & \qw & \qw & \qw & \qw & \qw & \qw & \qw & \qw & \gate{H} & \qw  & \qw & \qw & \qw  & \qw  \\
& \qw & \qw & \qw & \qw & \gate{H} & \qw & \qw  & \gate{H} & \gate{S} & \ctrl{2} & \ctrl{1} & \gate{H} & \qw & \gate{H} & \qw & \qw & \qw & \qw  \\
& \qw & \qw & \qw & \qw & \qw & \gate{H} & \qw & \gate{H}& \qw & \qw & \gate{Y} & \gate{Y} & \qw & \qw & \gate{H} & \qw & \qw & \qw  \\
& \qw & \qw & \qw & \qw & \qw & \qw & \gate{H} & \qw & \gate{S^{\dagger}} & \ctrl{-2} & \qw & \ctrl{-1} & \qw & \qw & \qw & \gate{H} & \qw & \qw  
}\\
\\
&
\raisebox{-4.2em}{\rule{12em}{0em}$=$\rule{1.6em}{0em}} 
\Qcircuit @C=.7em @R=.7em {
& & & \lstick{\ket{+}} & \ctrl{2} & \ctrl{4} & \qw & \qw & \qw & \ctrl{2} & \qw & \qw & \ctrl{4} & \ctrl{2} & \ctrl{4} & \measureD{X} \\
\push{\rule{0em}{1.3em}} & \qw & \qw & \qw & \qw & \qw & \qw & \qw & \qw & \qw & \qw & \qw & \qw & \qw & \qw & \qw & \qw  \\
& \qw & \qw & \qw & \gate{H}& \qw  & \gate{H} & \gate{S} & \ctrl{2} & \ctrl{-2} & \ctrl{1} & \gate{H} & \qw & \gate{H} & \qw & \qw & \qw \\
& \qw & \qw & \qw & \qw & \qw & \gate{H}& \qw & \qw & \qw & \gate{Y} & \gate{Y} & \qw & \qw & \qw & \qw & \qw \\
& \qw & \qw & \qw & \qw & \gate{H} & \qw & \gate{S^{\dagger}} & \ctrl{-2} & \qw & \qw & \ctrl{-1} & \control \qw & \qw & \gate{H} & \qw & \qw \\}
\end{tabular}
\caption{Circuits that detect whether two input states are in the
  subspace spanned by $\ket{H}\ket{-H}$ and $\ket{-H}\ket{H}$, when
  the input states are (a) unencoded and (b) encoded in the four-qubit
  code.  Corresponding blocks of the circuits in parts (a) and (b) are
  indicated by shading.  The translation from (a) to (b) relies on the
  fact that, for the four-qubit code, transversal Hadamard effects
  $\bar{H}_1 \bar{H}_2 \, \bar{\protect\SWAP}_{12}$.  The
  equivalence in b) eliminates four $\protect\CH$ gates, reducing
  the number of resource states required by eight.  In total,
  this figure shows that only four $\protect\CH$ gates are required to
  project the two qubits encoded in the four-qubit code into either the
  subspace spanned by $\{\ket{H}\ket{-H},\ket{-H}\ket{H}\}$ or that
  spanned by $\{\ket{H}\ket{H},\ket{-H}\ket{-H}\}$.  The
  circuits shown also apply an incidental Hadamard gate to the
  second logical qubit. Further details are given in
  Fig.~\ref{fig:CircuitDetails} in the
  appendix. \label{fig:logicalMeas}}
\end{figure*}
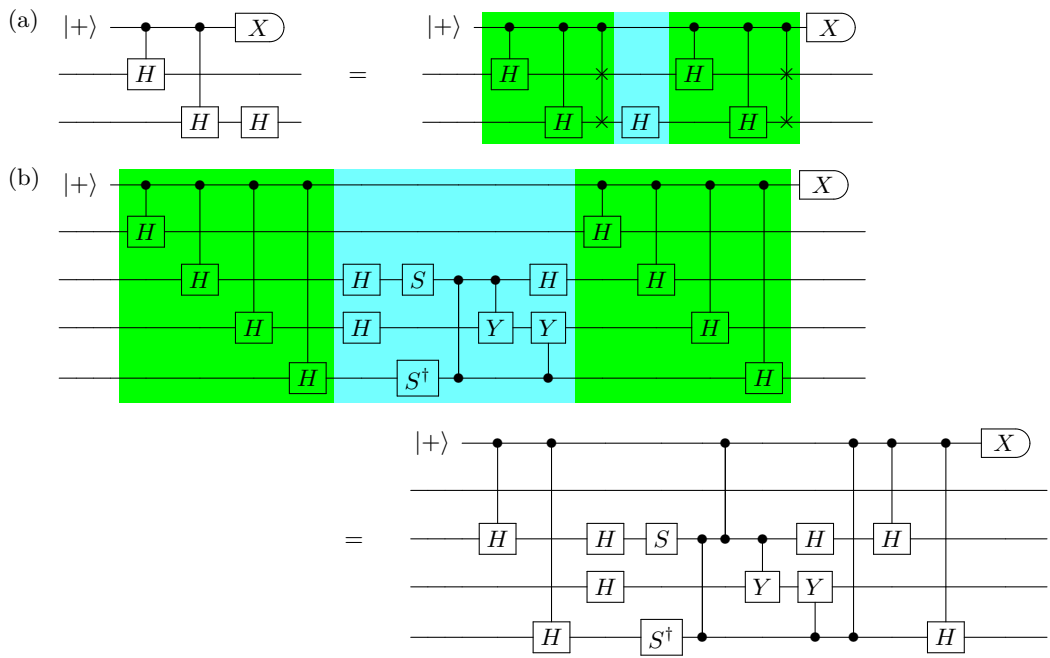

We denote the $+1$ ($-1$) eigenstate of the Hadamard gate by $\ket{H}$
($\ket{-H}$), where $\ket{H} = \cos(\frac{\pi}{8})\ket{0} +\,
\sin(\frac{\pi}{8})\ket{1} = Y(\frac{\pi}{4})\ket{0}$.  $\ket{H}$ is
not a stabilizer state, so Clifford operations are not sufficient for its preparation; 
however, they are sufficient for its distillation \cite{BravyiKitaev2005}.
As shown in Fig.~\ref{fig:littleCircuits}(c), $\ket{H}$ can be used together with
Clifford operations to implement the non-Clifford gate
$Y(\frac{\pi}{4})$, so $\ket{H}$ is a magic state.

Distillation is the process of converting multiple faulty copies of a
desired (resource) state into, typically fewer, improved copies of the
state.  In magic-state distillation, Clifford operations are used to
project faulty magic states into a subspace. Given input states of
suitable quality and a well-chosen subspace, successful projection
allows one to extract higher-fidelity copies of the desired state. When failure
to project is detected, the output states are discarded.  It is
assumed that Clifford operations can be implemented perfectly, which
is justified in the commonly considered situation where fault-tolerant
techniques provide highly accurate Clifford operations as a matter of
course but rely on techniques such as magic-state distillation for
universality.

State distillation is facilitated by randomization, which can be used
to simplify errors on resource states.  For example, the twirling
superoperator
\begin{align}
  \opfunc{H}(\rho) =\frac{1}{2} \rho + \frac{1}{2}H \rho H^\dagger \;,
\end{align}
which can be implemented by applying either $I$ or $H$ with equal
probability, decoheres an input state in the eigenbasis of the $H$
operator.  For any input state, the resulting state is a
probabilistic mixture of $\ket{H}$ and $\ket{-H}$ states. That is, for
some $0\leq p \leq 1$,
\begin{align*}
  \opfunc{H}(\rho) = (1-p) &\ket{H}\bra{H} + p \ket{-H}\bra{-H} \;.
\end{align*}
Because $Y \ket{H} = \ket{-H}$, we can characterize any faulty $\ket{H}$
state that has been twirled with $\opfunc{H}$ as suffering from
stochastic $Y$ errors with some probability $p$ that depends on the input 
state $\rho$.  We assume throughout
this paper that resource states are twirled prior to use.

We label distillation routines by their input/output ratios, so an
$m$-to-$n$ distillation routine takes $m$ resource states as input and
produces $n$ resource states as output.

\section{$10$-to-$2$ Distillation Routine\label{sec:HState}}

The basic form of our routine for magic-state distillation is as
follows: Resource states are encoded into a quantum code; these
encoded resource states are verified through an encoded measurement;
and finally the code is decoded, leaving, when no errors are
indicated, resource states of better quality.  The intuition behind
this approach is that one would like simply to measure whether a
resource state is good, but doing so requires additional resource
states whose own errors might go undetected in such a measurement.
Errors on these states are rendered detectable by performing an
encoded version of the measurement in a fault-tolerant fashion.  This
is the approach employed in reference~\cite{Knill2004} for distilling
$\ket{H}$ using the $7$-qubit Steane code. The routine described here
is instead based on the $4$-qubit error-detecting code.

As the $+1$ eigenstate of the Hadamard operator, the state $\ket{H}$
can be verified by measuring $H$.  Measurement of the Hadamard
operator is impossible using only Clifford operations, but it can be
accomplished, as shown in Fig.~\ref{fig:littleCircuits}, with the help
of two additional $\ket{H}$ states.

To render errors during the Hadamard measurement detectable, the
routine first encodes a pair of faulty resource states into the
$\code{C}_4$ code~\cite{Knill2004} and then performs an encoded measurement $\bar{H}_1
\bar{H}_2$ on the pair. This measurement determines
whether the pair is in the logical subspace spanned by
$\ket{H}\ket{-H}$ and $\ket{-H}\ket{H}$ and can therefore detect whether
one of the states had an error (see Sec.~\ref{sec:analysis}).

The $\code{C}_4$ code is a $[[4,2,2]]$ quantum code defined by the
stabilizer generator matrix:
\begin{align}
  \left[
    \begin{tabular}{c}
      $X\otimes X\otimes X\otimes X$\\
      $Z\otimes Z\otimes Z\otimes Z$
    \end{tabular}
  \right]\;.
\end{align}
Our choices for logical $X$ and $Z$ operators are:
\begin{align*}
  \bar{X}_1&=X\otimes X\otimes I\otimes I, \\
  \bar{Z}_1&=Z\otimes I\otimes I\otimes Z, \\
  \bar{X}_2&=X\otimes I\otimes I\otimes X, \;\;\; \text{and} \\
  \bar{Z}_2&=Z\otimes Z\otimes I\otimes I.
\end{align*}
Because any one-qubit Pauli operator anticommutes with some stabilizer
generator of $\code{C}_4$, it is possible to detect any error on a
single qubit of the code.

The set of stabilizer generators of $\code{C}_4$ is symmetric with
respect to exchange of $X$ and $Z$, so $H\otimes H\otimes H\otimes H$
is a valid encoded gate, and for the choice of logical Pauli operators
given above it effects a logical Hadamard on both encoded qubits
followed (or, equivalently, preceded) by a logical \classicallogicgate{SWAP}.
Consequently, the controlled-$(\bar{H}_1 \bar{H}_2 \, \bar{\SWAP}_{12})$ gate (the
control is unencoded and the target is encoded in $\code{C}_4$) can be
accomplished using a sequence of four $\CH$ gates.  Using this gate
one can derive a circuit that implements the encoded measurement, $\bar{H}_1\bar{H}_2$, as shown in Fig.~\ref{fig:logicalMeas}, by means of four
$\CH$ gates implemented with a total of eight resource states.

\begin{figure}
\begin{tabular}{ll}
(a) & 
\Qcircuit @C=.7em @R=.7em {
 & & & \lstick{\ket{H}} & \gate{Y} & \ctrl{1} & \measureD{\pm Y} \cwx[1] \\
 & \qw & \qw & \qw & \qw & \gate{Y} & \gate{Y(\frac{\pm \pi}{2})} & \qw }
\raisebox{-1.3em}{\rule{.8em}{0em}$=$\rule{.5em}{0em}}
\Qcircuit @C=.7em @R=.7em {
& & & \lstick{\ket{H}} & \ctrl{1} & \measureD{\pm Y} \cwx[1] \\
& \qw & \qw & \qw & \gate{Y} & \gate{Y(\frac{\pm \pi}{2})} & \gate{Y} & \qw
} \\ \\ \\
(b) &
    \Qcircuit @C=.7em @R=.7em {
      \push{\rule{0em}{1.3em}} & \qw & \qw & \ctrl{1} & \qw & \qw \\
      & \gate{Y(\frac{-\pi}{4})} & \gate{Y} & \control \qw & \gate{Y(\frac{\pi}{4})} & \qw
    }
    \raisebox{-1.3em}{\rule{.8em}{0em}$=$\rule{.5em}{0em}}
    \Qcircuit @C=.7em @R=.7em {
      & \qw & \ctrl{1} & \qw & \gate{Z}  & \qw \\
      & \gate{Y(\frac{-\pi}{4})} & \control \qw & \gate{Y(\frac{\pi}{4})} & \gate{Y} & \qw
    }
    \\ \\
\end{tabular}
 \caption{Circuit identities for propagating $Y$ 
   errors on gate states. (a) The rule
   for propagating a $Y$ error on a resource state used to apply a
   $Y(\frac{\pm\pi}{4})$ gate as in Fig.~\ref{fig:littleCircuits}(c). 
   The effect of the error is the same as
   applying the $Y$ error after the gate.  (b) A $Y$ error on the first
   resource state required to implement a $\protect\CH$ gate using the
   circuits in Fig.~\ref{fig:littleCircuits} propagates to a $Z$ error
   on the control and a $Y$ error on the target.  An error on the
   second resource state propagates trivially to a $Y$ error on the
   target.\label{fig:Yerror}}
\end{figure}
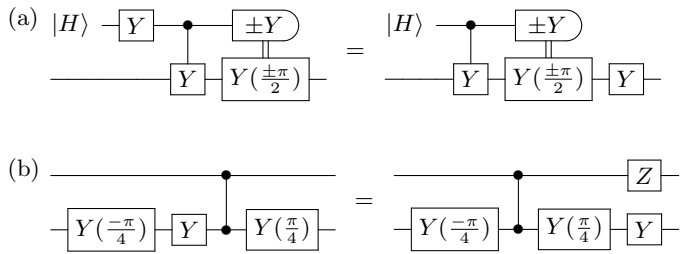

The final step of the distillation routine is to use Clifford
operations to decode the logical qubits and measure the syndrome of
the $\code{C}_4$ code, leaving two output resource states.  The
routine succeeds and accepts the output if neither the encoded
measurement nor the syndrome indicates an error. Otherwise the output
is discarded.  We analyze the error patterns for the full distillation
circuit, shown in Fig.~\ref{fig:errorAnalysis}(a), in the next section.

\section{Analysis\label{sec:analysis}}

Given perfect Clifford operations and twirled resource states, the
only possible errors in our distillation circuit are $Y$ errors on the
input resource states.  For simplicity we assume that the input states
to be distilled are independent and all have the same error
probability $p$.

As described in the previous section, the ten input resource states
can be partitioned into two resource states that are encoded into the
code $\code{C}_4$ (data states) and four pairs of resource states used
to implement $\CH$ gates (gate states).  The effect of one error
on either type of resource state can be understood as follows.

A $Y$ error on one of the data states becomes an encoded $Y$ error,
which flips the outcome of the encoded measurement (the measurement of $\bar{H}_1 \bar{H}_2$)
and is thus detected by the routine.  The decoding exactly
reverses the encoding, and the logical gates in between preserve
logical $Y$ errors, so errors on data states persist on the output and
are not detected by the syndrome measurement.

As shown in Fig.~\ref{fig:Yerror}, a $Y$ error on one of the gate
states causes the intended $\CH$ gate to act as $\CH$
followed by either $Z\otimes Y$ or $I\otimes Y$, depending on which
resource state was in error.  Using circuit identities, these errors
can be propagated to a common location just before the second set of
$\CH$ gates, as depicted in Fig.~\ref{fig:errorAnalysis}(b).  At this
location, such an error appears as a combination of some
logical operator and a $Y$ error on a single qubit, which is not an
encoded Pauli operator for the $\code{C}_4$ code.  This $Y$ error is
followed only by logical operators, which cannot take an error
subspace to a non-error subspace, and decoding, which returns a
syndrome indicating whether the state is in an error subspace.
Consequently, a single error on a gate state is detected by
the syndrome measurement.

\begin{figure*}
\begin{tabular}{ll}
(a) & 
\begin{pgfpicture}{0cm}{0cm}{0cm}{0cm}
\color{yellow}
\pgfrect[fill]{\pgfpoint{0cm}{-.25cm}}{\pgfpoint{.65cm}{-.6cm}}
\pgfrect[fill]{\pgfpoint{0cm}{-1.45cm}}{\pgfpoint{.65cm}{-.6cm}}
\pgfrect[fill]{\pgfpoint{2.57cm}{.2cm}}{\pgfpoint{1.4cm}{-2.9cm}}
\pgfrect[fill]{\pgfpoint{7.85cm}{.2cm}}{\pgfpoint{1.4cm}{-2.9cm}}
\end{pgfpicture}
\Qcircuit @C=.7em @R=.7em {
\push{\rule{1.3em}{0em}} & & & & & \lstick{\ket{+}} & \ctrl{2} & \ctrl{4} & \qw & \qw & \qw & \ctrl{2} & \qw & \qw & \ctrl{4} & \ctrl{2} & \ctrl{4} & \measureD{X} \\
& \lstick{\ket{H}} & \targ & \qw & \ctrl{1} & \qw & \qw & \qw & \qw & \qw & \qw & \qw & \qw & \qw & \qw & \qw & \qw & \ctrl{1} & \qw & \targ & \qw \\
& \lstick{\ket{0}} & \qw & \targ  & \targ & \qw & \gate{H} & \qw & \gate{H} & \gate{S} & \ctrl{2} & \ctrl{-2} & \ctrl{1} & \gate{H} & \qw & \gate{H} & \qw & \targ & \targ & \qw & \measureD{Z} \\
& \lstick{\ket{H}} & \qw & \ctrl{-1} & \targ & \qw & \qw & \qw & \gate{H} & \qw & \qw & \qw & \gate{Y} & \gate{Y} & \qw & \qw & \qw & \targ & \ctrl{-1} & \qw & \qw \\
& \lstick{\ket{+}} & \ctrl{-3} & \qw & \ctrl{-1} & \qw & \qw & \gate{H} & \qw & \gate{S^{\dagger}} & \ctrl{-2} & \qw & \qw & \ctrl{-1} & \ctrl{-4} & \qw & \gate{H} & \ctrl{-1}& \qw & \ctrl{-3} & \measureD{X} \\
}
\\
\\
(b) &
\begin{pgfpicture}{0cm}{0cm}{0cm}{0cm}
\color{yellow}
\pgfrect[fill]{\pgfpoint{2.468cm}{-.475cm}}{\pgfpoint{1.318cm}{-2.248cm}}
\pgfrect[fill]{\pgfpoint{10.32cm}{.148cm}}{\pgfpoint{1.117cm}{-2.862cm}}
\end{pgfpicture}
\Qcircuit @C=.7em @R=.7em {
\push{\rule{1.3em}{0em}} & & & & & & \lstick{\ket{+}} & \ctrl{1} & \ctrl{2} & \ctrl{3} & \ctrl{4} & \qw & \qw & \qw & \qw & \qw & \multigate{4}{\parbox{.8cm}{Pauli\\errors}} & \ctrl{1} & \ctrl{2} & \ctrl{3} & \ctrl{4} & \measureD{X} \\
& \lstick{\ket{H}} & \targ & \qw & \ctrl{1} & \multigate{3}{\parbox{1.0cm}{Logical\\$Y$\\errors}} & \qw & \gate{H} & \qw & \qw & \qw & \qw & \qw & \qw & \qw & \qw & \ghost{\parbox{.8cm}{Pauli\\errors}} & \gate{H} & \qw  & \qw & \qw & \ctrl{1} & \qw & \targ & \qw \\
& \lstick{\ket{0}} & \qw & \targ  & \targ & \ghost{\parbox{1.0cm}{Logical\\$Y$\\errors}} & \qw & \qw & \gate{H} & \qw & \qw  & \gate{H} & \gate{S} & \ctrl{2} & \ctrl{1} & \gate{H} & \ghost{\parbox{.8cm}{Pauli\\errors}} & \qw & \gate{H} & \qw & \qw & \targ & \targ & \qw & \measureD{Z} \\
& \lstick{\ket{H}} & \qw & \ctrl{-1} & \targ & \ghost{\parbox{1.0cm}{Logical\\$Y$\\errors}} & \qw & \qw & \qw & \gate{H} & \qw & \gate{H}& \qw & \qw & \gate{Y} & \gate{Y} & \ghost{\parbox{.8cm}{Pauli\\errors}} & \qw & \qw & \gate{H} & \qw & \targ & \ctrl{-1} & \qw & \qw \\
& \lstick{\ket{+}} & \ctrl{-3} & \qw & \ctrl{-1} & \ghost{\parbox{1.0cm}{Logical\\$Y$\\errors}} & \qw & \qw & \qw & \qw & \gate{H} & \qw & \gate{S^{\dagger}} & \ctrl{-2} & \qw & \ctrl{-1} & \ghost{\parbox{.8cm}{Pauli\\errors}} & \qw & \qw & \qw & \gate{H} & \ctrl{-1}& \qw & \ctrl{-3} & \measureD{X} \\}
\\
\\
(c) &
\begin{pgfpicture}{0cm}{0cm}{0cm}{0cm}
\color{red}
\pgfrect[fill]{\pgfpoint{.987cm}{-.476cm}}{\pgfpoint{1.117cm}{-.931cm}}
\pgfrect[fill]{\pgfpoint{5.028cm}{.15cm}}{\pgfpoint{1.111cm}{-1.551cm}}
\end{pgfpicture}
\Qcircuit @C=.7em @R=.7em {
\push{\rule{1.3em}{0em}} & & & \lstick{\ket{+}} & \ctrl{1} & \ctrl{2} & \ctrl{1} & \qw & \multigate{2}{\parbox{.8cm}{Pauli\\errors}} & \ctrl{1} & \ctrl{2} & \ctrl{1} & \measureD{X} \\
& \lstick{\ket{H}} & \multigate{1}{\parbox{.8cm}{$Y$\\errors}} & \qw & \gate{H} & \qw & \qswap \qwx[1] & \qw  & \ghost{\parbox{.8cm}{Pauli\\errors}} & \gate{H} & \qw & \qswap \qwx[1] & \qw & \qw \\
& \lstick{\ket{H}} & \ghost{\parbox{.8cm}{$Y$\\errors}} & \qw & \qw & \gate{H} & \qswap & \gate{H} & \ghost{\parbox{.8cm}{Pauli\\errors}} & \qw & \gate{H} & \qswap & \qw & \qw \\
}
\end{tabular}
\caption{Illustration of a method of classifying errors in our
  distillation routine.  The first circuit shows the full distillation
  routine with possible error locations shaded.  Using circuit
  identities, errors at any of these locations can be concentrated
  into one of two regions (and types) yielding an equivalent circuit
  of the form shown in (b).  Any Pauli errors on the lower four (code)
  qubits in the second error region of (b) will be detected by the
  decoding circuit unless they act as a logical (Pauli) operator on
  the code.  The remaining possible errors, those undetectable by the
  syndrome measurement, may then be enumerated and classified efficiently
  using the logical circuit shown in (c).\label{fig:errorAnalysis} }
\end{figure*}
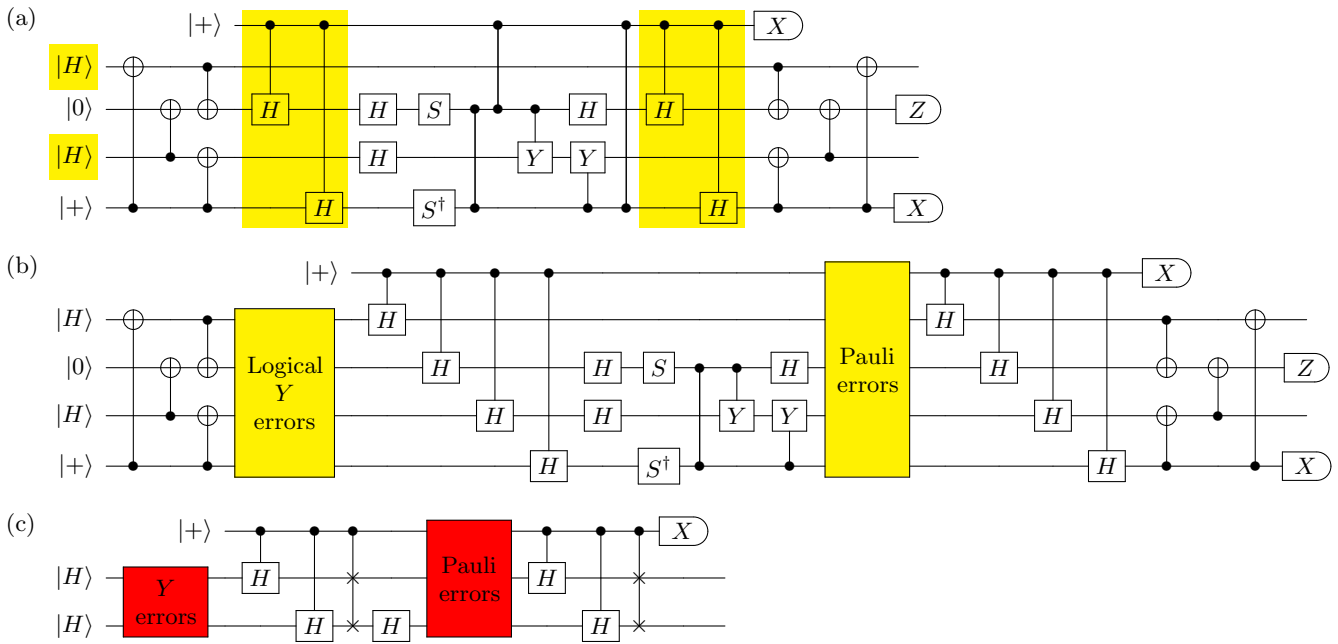

The effect of multiple errors is best understood by propagating the
errors from both gate and data states to two locations, as described
in Fig.~\ref{fig:errorAnalysis}.  The $Y$ Pauli operators from any pair of errors
on gate states (described above) combine to form a logical operator for the code, 
so any even number
of errors on such states will fail to be detected by the decoder, while
any odd number of errors will be detected.  For each error
pattern that is not detected by the syndrome, 
one can consider the effect of the logical errors on the
encoded information and encoded measurement.  For example, 
even numbers of $Z$ errors on the encoded-measurement ancilla 
will cancel and cause the 
distillation to be accepted.  Each pattern of errors on the resource states
can then be classified first by whether it is detected and then by
whether it causes a non-trivial logical error. Because errors on each
state are considered to be equiprobable and independent, this
enumeration determines the probability $a(p)$ of the distillation
routine accepting and the marginal error probability $e(p)$ of an
output state conditional on acceptance.  It happens that $e(p)$
does not depend on which of the two output states is considered.

Based on the observation that any single error results in rejection, a simple
estimate of the acceptance probablity is $a(p) = 1 - 10 p + O(p^2)$.  The exact
accounting yields
\begin{align*}
  a(p) = 1 &- 10 p + 58 p^2 - 192 p^3 + 400 p^4\\ &- 544 p^5 + 480 p^6
  - 256 p^7 + 64 p^8\;.
\end{align*}

The probability of an undetected error on the output states is the
probability that the routine accepts and that the output nevertheless
has an error. Because any single error is detected, this probability
has order $p^2$.  The marginal undetected-error probability of the
first (or identically the second) output state is
\begin{align*}
  u(p) = 9 p^2 &- 56 p^3 + 160 p^4 - 256 p^5\\
  &+ 240 p^6 - 128 p^7 + 32 p^8\;.
\end{align*}
For our purposes, this is the quantity of interest, but one can also
compute the probability of at least one undetected error on the two outputs.
This is given by
\begin{align*}
  u_2(p) = 13 p^2 &- 80 p^3 + 228 p^4 - 368 p^5\\
  &+ 352 p^6 - 192 p^7 + 48 p^8\;.  
\end{align*}

The quality of the distillation routine's output is quantified by the
marginal probability of error of an output state conditional on
acceptance:
\begin{align}
  e(p) = u(p)/a(p) \;.\label{eq:po}
\end{align}
The corresponding probability of at least one error on
the two outputs conditioned on acceptance is $e_2(p) = u_2(p)/a(p)$.
It can be shown numerically that $e_2(p)\leq 2 e(p)- e(p)^2$, so
errors on the two output states are positively correlated.  In fact,
the probability of an error on both output states is of order
$p^2$.

\section{Distillation Sequences\label{sec:distillationSequences}}

The ultimate goal of magic-state distillation is to produce resource
states of sufficiently high quality that they can be used to
implement all non-Clifford gates in a computation without
significantly increasing the probability that the computation will
fail.  A generic computation will fail if any single gate fails, so
the probability of one or more errors on the $R$ resource states
employed in a computation must be much less than $1$ to ensure that
the computation succeeds with high probability.  By the union bound,
it is sufficient that the marginal probability of error on each
resource state be much less than $1/R$.  Strong correlations can
reduce this requirement on marginal probabilities, but for independent
errors the bound is necessary.  Consequently, the proximate goal of
magic-state distillation is to produce resource states such that the
marginal probability of error for any single state is bounded from
above by some goal error probability, $e_g\ll 1/R$.  In algorithms
currently envisoned for quantum computers, $R$ can easily be
$10^{10}$ or more.

In order to obtain resource states with very low probabilities of
error, it is necessary to use multiple rounds of distillation, where
the input to each round is produced by the preceding one.  We consider
a sequence of such rounds where each is based on a single but possibly
round-dependent distillation routine.  In a round based on an
$m$-to-$n$ distillation routine, the output resource states from the
preceding round are grouped into blocks of size $m$, and each block is
then distilled to $n$ states, which may, in general, have correlated errors.

The sequence of rounds is chosen to minimize the number of input
resource states needed to produce a given number of output states with
marginal probability of error $e_g$ or less.  In practice, we are
interested in the case where the number of resource states to be
prepared is very large, allowing us to consider only the asymptotic
cost. The cost is defined as the number of input resource states used
per output resource state produced.  For one round of the $10$-to-$2$
distillation routine, the cost is $\frac{10}{2 a(p)}$, with marginal
probability of error $e(p)$ on the output states conditioned on
acceptance.

The one-round expression for marginal probability of error given in
Sec.~\ref{sec:analysis} assumes that the input resource states suffer
from errors independently and with equal probability.  Generally,
however, the output of a distillation routine need not satisfy either
restriction, which poses a concern for distillation sequences
involving multiple rounds.  If necessary, distillation routines can be
output symmetrized by randomly permuting the output states, thereby
ensuring that the output states from a given round all have the same
error probability.  Independence is a concern whenever a routine that
outputs more than one state per instance is used, since errors on the
states output by one instance of such a routine are usually not
independent.  For example, the probability of two errors in the output
of the $10$-to-$2$ distillation routine is of the same order as that
for one error.  Performing a distillation routine using such correlated states
as input can substantially increase the output error probability.  To
avoid this effect, it is sufficient to ensure that no instance of a
routine depends on more than one output from any previously executed
instance of a routine.  The following lemma and its corollary show
that this strategy works without an increase in asymptotic cost. As a
consequence we can calculate the asymptotic cost as if the output
states of all routines were completely independent.

\begin{lemma}
\label{thm:one_step_rate}
Let $\cD$ be an $m$-to-$n$ output-symmetrized distillation routine
with acceptance probability $a(p)$ and conditional output error
probability $e(p)$ for each output state. Given a block of $K$
independent resource states, each with error probability $p$, one can
produce $n$ blocks of output states where each block's states are
independent within the block and have probability of error $e(p)$.
Each block contains $a(p) \lfloor K/m\rfloor$ states on average.
\end{lemma}

\begin{proof}
Partition the $K$ resource states into $\lfloor K/m\rfloor$ sets of
$m$ states, discarding any remaining ones.  Apply $\cD$ to each
set of $m$ states, getting $n$ states with probability 
$a(p)$ in each case. Conditional on acceptance, each output state has
marginal probability of error $e(p)$, though these errors are not
independent. Form $n$ blocks by taking the $j^{\textrm{th}}$ output
state from each successful distillation, for $j=1,\ldots, n$. These
blocks have the desired properties for a given pattern of distillation
successes. Because the error probabilities of the $j^{\textrm{th}}$
output states of the successful distillations do not depend on the
pattern of successes, any such arrangement of the output states that
depends only on the pattern of successes preserves this independence.
\end{proof}

\begin{figure}
\includegraphics[width=8.5cm]{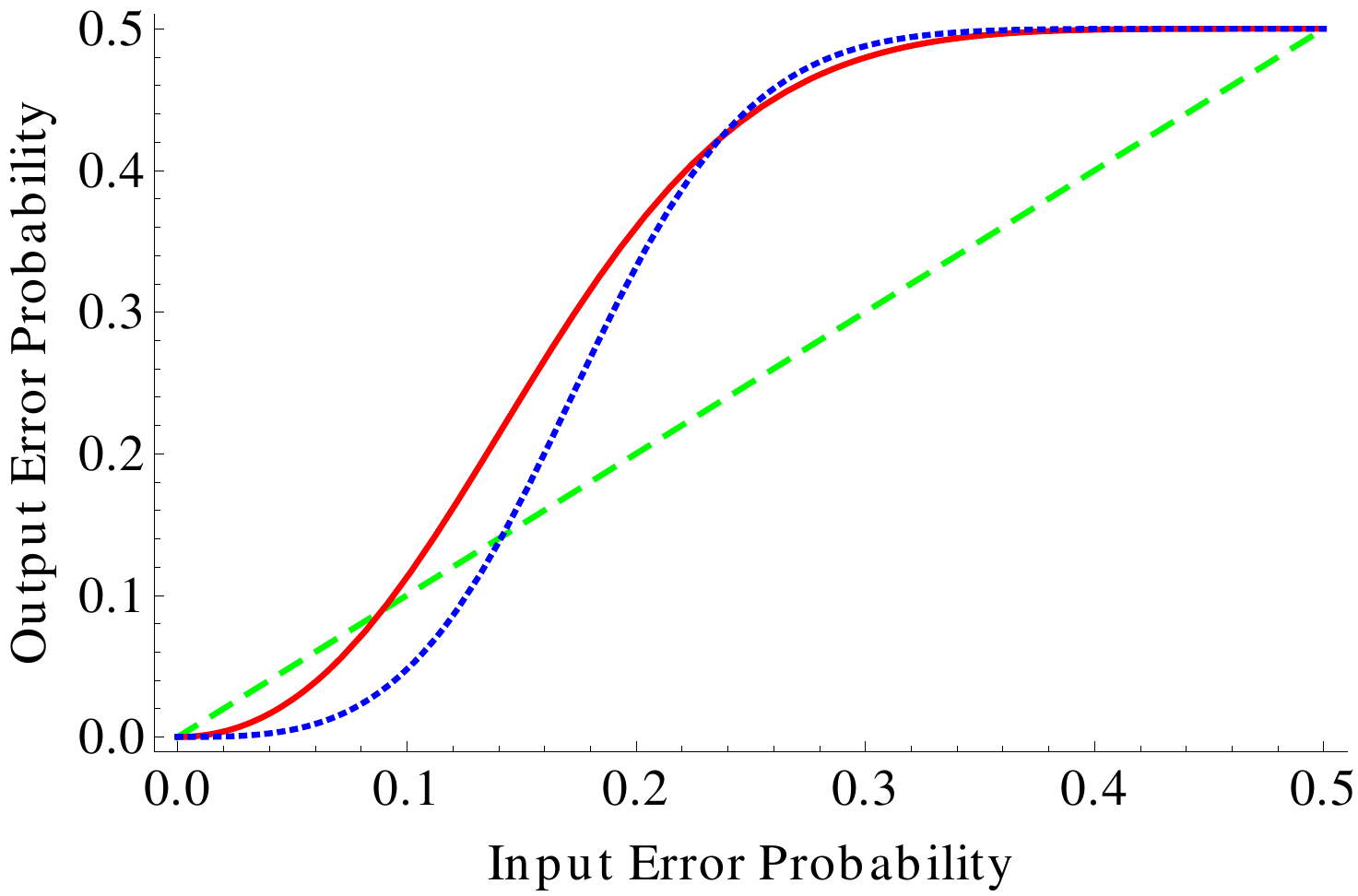}
\caption{Plots of the marginal error probabilities conditional on
  acceptance for the $10$-to-$2$ (solid) and $15$-to-$1$ (dotted)
  routines.  The dashed line indicates the output if no distillation
  is performed. The thresholds for the two routines are determined by
  the first intersections with this line.
\label{fig:both_thresh}}
\end{figure}

\begin{corollary}
\label{cor:one_step_rate}
If, in Lem.~\ref{thm:one_step_rate}, $K$ is random with average
$\langle K \rangle$, then the expected total number of output states is at
least $a(p) n \left (\frac{\langle K \rangle}{m}-1 \right )$. The average size of each of
the $n$ output blocks of independent states is at least
$a(p) \left (\frac{\langle K \rangle}{m}-1 \right )$.
\end{corollary}

A multi-round distillation routine can now be formulated as follows:
Assume that after round $l-1$ there are $N_{l-1}$ blocks of resource states,
where within each block the states are independent with identitical
error probabilities $p_{l-1}$, and the number of states in each block is
$\langle K_{l-1} \rangle$ on average. Applying the procedure of
Lem.~\ref{thm:one_step_rate} to each block with a $m_{l}\rightarrow n_{l}$
distillation routine $\cD_{l}$ yields $N_{l}=N_{l-1}n_{l}$ output blocks,
where each output block has $\langle K_{l} \rangle \geq a_{l}(p_{l-1})
\left (\frac{\langle K_{l-1} \rangle}{m_{l}}-1 \right )$ resource states on average, 
independent within a block
and each with error probability $p_{l}=e_{l}(p_{l-1})$. The first round
starts with $K_0$ independent resource states, each of which suffer an
error with probability $p_0$. For large $K_0$, the constant offsets of
$-1$ in the expressions are negligible.  Consequently, the asymptotic
cost $c_{l}$ of resource-state production after round $l$ satisfies $c_{l} = \frac{m_{l}}{n_{l} a_{l}(p_{l-1})}c_{l-1}$, where $c_0=1$.  The error
probability after round $l$ satisfies $p_{l} = e_{l}(p_{l-1})$.

\section{Comparative Performance \label{sec:comparativePerformance}}

At present, practical fault-tolerant architectures require
physical-gate error probabilities well below $.01$, which suggests
that it should be possible to directly prepare resource states with
error probabilities of a few percent or less.  Given such states, the
most practical routine for $\ket{H}$ distillation developed to date is
the $15$-to-$1$ routine.  In this section, we compare the performance
of the $10$-to-$2$ routine to that of the $15$-to-$1$ routine and
consider the effect of using them in concert.

\begin{figure*}
\includegraphics[width=15cm]{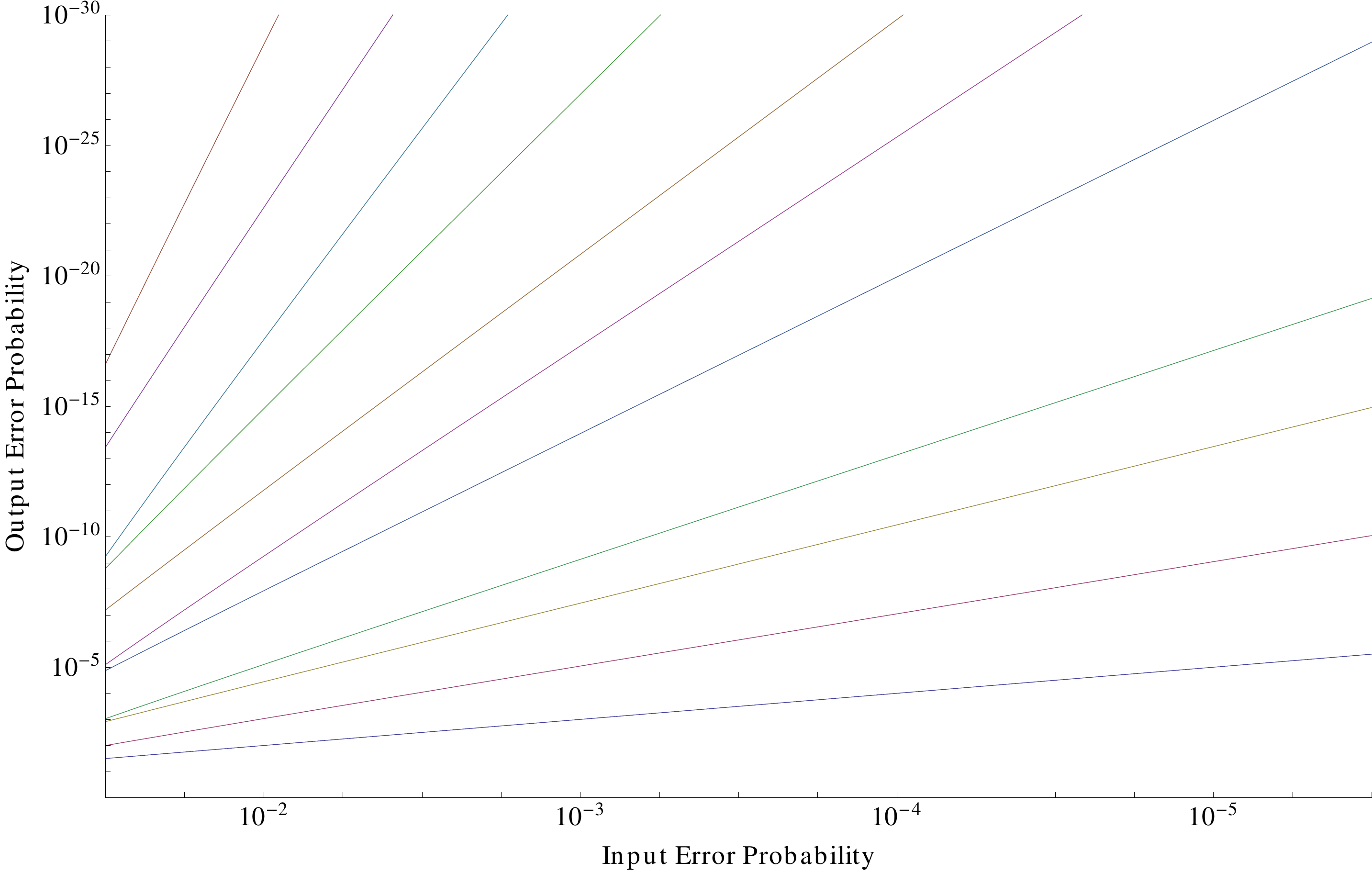}
\begin{pgfpicture}{0cm}{0cm}{0cm}{0cm}
\pgfputat{\pgfxy(-1.1,2)}{\pgfbox[center,center]{No distillation}}
\pgfputat{\pgfxy(-.4,3.4)}{\pgfbox[center,center]{A}}
\pgfputat{\pgfxy(-.4,4.8)}{\pgfbox[center,center]{B}}
\pgfputat{\pgfxy(-.4,5.9)}{\pgfbox[center,center]{AA}}
\pgfputat{\pgfxy(-.4,8.5)}{\pgfbox[center,center]{BA}}
\pgfputat{\pgfxy(-2.8,9.25)}{\pgfbox[center,center]{AAA}}
\pgfputat{\pgfxy(-4.85,9.25)}{\pgfbox[center,center]{BB}}
\pgfputat{\pgfxy(-7.4,9.25)}{\pgfbox[center,center]{BAA}}
\pgfputat{\pgfxy(-9.0,9.25)}{\pgfbox[center,center]{AAAA}}
\pgfputat{\pgfxy(-10.35,9.25)}{\pgfbox[center,center]{BBA}}
\pgfputat{\pgfxy(-11.5,9.25)}{\pgfbox[center,center]{BAAA}}
\end{pgfpicture}
\caption{Output error probability as a function of input error
  probability for various sequences of the $10$-to-$2$ (A) and
  $15$-to-$1$ (B) $\ket{H}$ distillation routines.  Each curve is
  labeled by the associated sequence of routines, e.g., $BAA$ denotes
  the $15$-to-$1$ routine followed by two rounds of the $10$-to-$2$
  routine.  The region directly underneath each labeled curve is the
  region in which the labeled strategy is preferred.
  \label{fig:regionplot}}
\end{figure*}

Routines for magic-state distillation are typically judged on the
basis of their threshold, that is, the error probability below which
resource states can be successfully distilled.  At the threshold, a
distillation routine outputs resource states no better than the
inputs.  Thus, the threshold $p_t$ for the $10$-to-$2$ routine can
be determined from Eq.~\eqref{eq:po} by considering solutions to
$p_t=e(p_t)$.  This yields a threshold of $p_t = 0.089$, which is
substantially below the threshold of $0.141$ for the $15$-to-$1$
routine~\cite{BravyiKitaev2005}, but either threshold should be
adequate for the error regime of interest.  The curves for the
marginal output error probability of the $10$-to-$2$ and $15$-to-$1$
routines are plotted in Fig.~\ref{fig:both_thresh}.

The efficiency of a distillation routine can be
characterized, as detailed in Ref.~\cite{BravyiKitaev2005},
by the output error probability as a function of the number 
of resource states employed.  In the limit of small initial error 
probability $p$, the output error
probability for the $10$-to-$2$ routine after $l$ rounds of 
distillation is $\frac{1}{9}(9 p)^{2^l}$.  
In the limit of both small $p$ and many output states, $l$ rounds of
distillation require $k=5^l$ input resource states per output.  Consequently,
taking $l$ to be continuous, the asymptotic output error probability as a 
function of the number of input resource states expended is 
$\frac{1}{9}(9 p)^{k^\xi},$ where
$\xi = \frac{1}{\log_2(5)} \approx .43$. 
The corresponding exponent for the $15$-to-$1$ routine is $.4$, 
so the $10$-to-$2$ routine performs slightly better for this metric. 
However, these smooth functions hide the step discontinuities induced by 
using sequences of increasing integral lengths (as seen in 
Fig.~\ref{fig:distplot}) and can be misleading for practical comparisons.

Of greater utility to us is the cost, in resource states consumed per
output state, required to obtain resource states of sufficiently high
quality for useful quantum computations, given resource states with
error probabilities in the range of $0.01$ to $10^{-5}$.  The cost
depends on the distillation sequence, which generally entails multiple
rounds of distillation.  For the purpose of optimizing the
distillation sequence, one can consider arbitrary routines at each
round.  Here, we consider sequences involving the $10$-to-$2$ and
$15$-to-$1$ distillation routines.

In Fig.~\ref{fig:regionplot} we plot the output error probability as a
function of input error probability for various sequences.  Data for
the $15$-to-$1$ routine was computed using the expressions
corresponding to $a(p)$ and $e(p)$ in Eq.~(35) and Eq.~(36) of
Ref.~\cite{BravyiKitaev2005}.  In the region plotted, distillation
sequences with higher output error (lower curves) also require fewer
input resource states per output state. Consequently, for a given
output error goal, $e_g$, and input error probability $p$, the label
of the nearest curve above the point $(p,e_g)$ in the plot gives the
best distillation sequence involving the $10$-to-$2$ and/or
$15$-to-$1$ routines.

\begin{figure*}
\begin{center}
\begin{picture}(400,220)(0,10)
\put(0,130){
\makebox(0,0)[l]{\includegraphics[width=13cm]{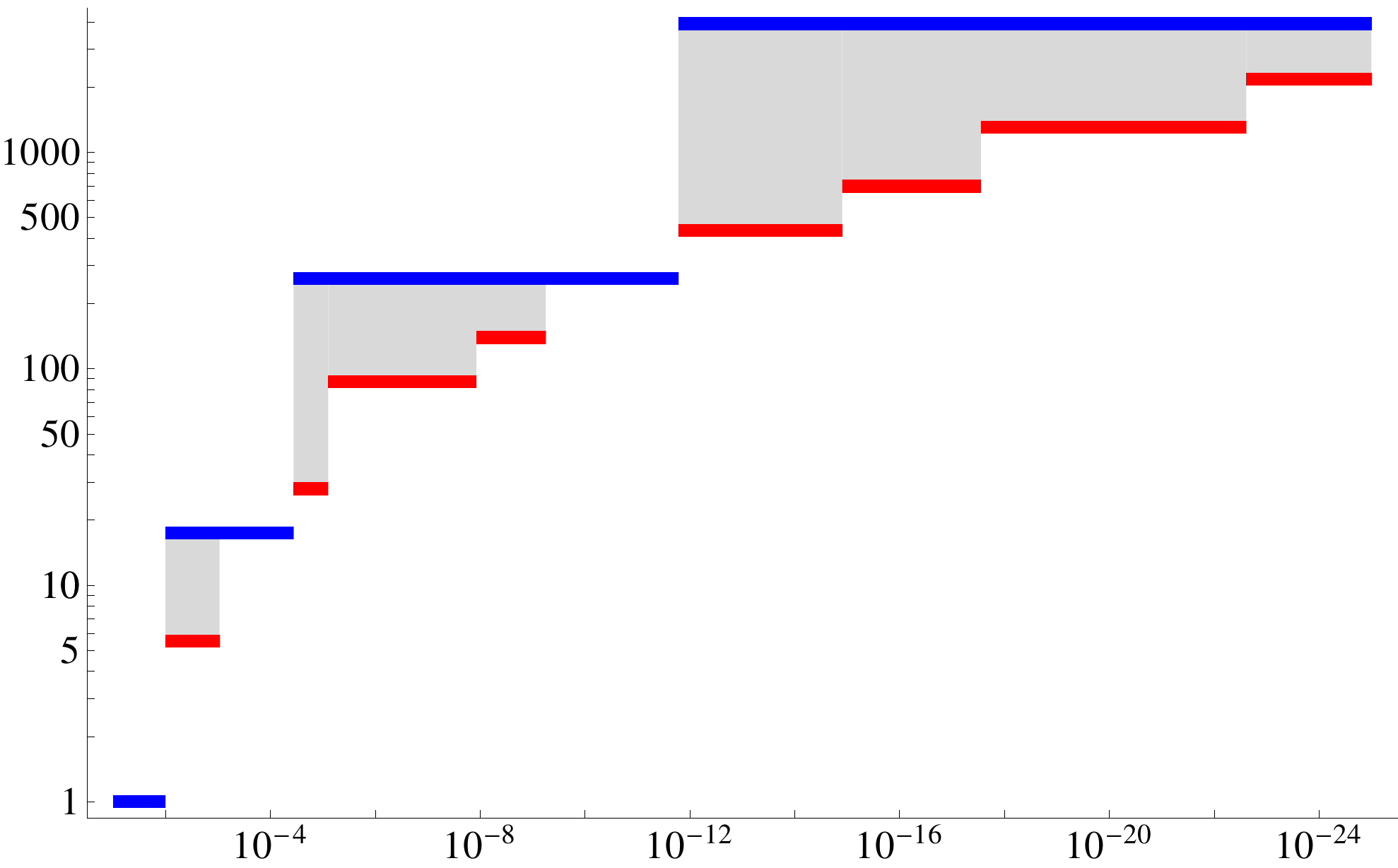}}}
\put(-10,110){\begin{sideways}Production Cost\end{sideways}}
\put(160,0){Goal Error Probability}
\end{picture}
\end{center}
\caption{Log-log plot of the cost required to produce output states
  satisfying a goal error probability ($e_g$) given input states with
  error probability $0.01$. The upper horizontal segments show the
  cost for the best sequence involving only the $15$-to-$1$ routine.
  The lower horizontal segments show the cost for the best sequence of
  routines that achieves $e_g$ or better. The gray regions indicate
  the improvement obtained using the
  $10$-to-$2$ routine.
  \label{fig:distplot}}
\end{figure*}

Table~\ref{tab:smallres} shows the costs and improvements for a number
of distillation sequences given an initial error probability of
$p=0.01$.  Not surprisingly, the table shows that the $10$-to-$2$
routine has a smaller cost, but the $15$-to-$1$ routine has greater
improvement in error probability per round.  For distillations that
use both routines, we find numerically that if $15$-to-$1$ rounds are
used they should be placed first.  This is intuitively consistent with 
the higher threshold for the $15$-to-$1$ routine, which suggests
better performance at high error probabilities. 

The cost improvements shown in Tab.~\ref{tab:smallres} are illustrated
more visually in Fig.~\ref{fig:distplot}, which shows the production
cost, at a fixed input error probability $p=0.01$ and as a function of
$e_g$, of the best distillation sequence compared to the best sequence
using only the $15$-to-$1$ routine.  For example, a goal error
probability near $10^{-5}$ can be achieved by using either two rounds
of the $15$-to-$1$ routine at a production cost of $261.7$ or two
rounds of the $10$-to-$2$ routine at a cost of $27.9$.  In this case,
the improvement in production cost obtained by incorporating the
$10$-to-$2$ routine is a factor of $9.4$.

\begin{table}
\begin{ruledtabular}
\begin{tabular}{cccc}
Distillation & Cost & Output error & Cost improvement \\
scheme & & probability, $e(p)$ & factor \\
\hline
$A$ & $5.5$ & $9\times10^{-4}$ & $3.2$ \\
$B$ & $17.4$ & $4\times10^{-5}$ & $1$ \\
$AA$ & $27.9$ & $7\times10^{-6}$ & $9.4$ \\
$BA$ & $87.2$ & $1\times10^{-8}$ & $3.0$ \\
$AAA$ & $139.3$ & $5\times10^{-10}$ & $1.9$ \\
$BB$ & $261.7$ & $2\times10^{-12}$ & $1$ \\
$BAA$ & $436.2$ & $1\times10^{-15}$ & $9.0$ \\
$AAAA$ & $696.6$ & $2\times10^{-18}$ & $5.6$ \\
$BBA$ & $1308.7$ & $2\times10^{-23}$ & $3.0$ \\
$BAAA$ & $2180.8$ & $1\times10^{-29}$ & $1.8$\\
\end{tabular}
\end{ruledtabular}
\caption{Costs and output error probabilities at $p=0.01$.  The labels
  for the distillation schemes follow the convention given in
  Fig.~\ref{fig:regionplot}.  The cost improvement factor is with respect
  to the shortest sequence using only the $15$-to-$1$ routine that
  achieves at least as good an output error probability.  
  \label{tab:smallres}}
\end{table}

\section{Conclusions \label{sec:conclusions}}

Magic-state distillation enables universal quantum computing given
only mediocre copies of a non-stabilizer state and high-quality
Clifford operations.  Considering the importance of Clifford-based
techniques to the theory of fault tolerance, we expect that
magic-state distillation will prove valuable for the practical
implementation of quantum computers.

At the logical level, computationally useful quantum algorithms
involve many non-Clifford gates, generally enough to account for a
significant fraction of all gates employed.  At least one high-quality
magic state is required for the indirect implementation of each
non-Clifford gate, so it is important to minimize the resources needed
for the distillation of such states.

In this work, we contributed to the goal of resource reduction by
introducing an $\ket{H}$ distillation routine that reduces the error
probability for faulty $\ket{H}$ states from $p$ to $O(p^2)$ and
produces $2$ output states using $10$ input states.  By judiciously
combining this routine with the higher-order ($p$ to $O(p^3)$) but higher-cost
$15$-to-$1$ routine from Refs.~\cite{BravyiKitaev2005,Knill2004}, we
showed that the number of faulty $\ket{H}$ states required to distill
states of a given quality can be reduced by up to an order of
magnitude.  Inclusion of additional distillation routines in the analysis 
would likely lead to further improvements.

\vspace*{.2in}

\begin{acknowledgments}
We thank Scott Glancy for his help in bringing this work to fruition.
This paper is a contribution of the National Institute of Standards
and Technology and not subject to U.S. copyright.
\end{acknowledgments}

\bibliography{DistBib}

\appendix*

\section{}

\begin{figure*}[ht]
\begin{tabular}{ll}
a) &
\Qcircuit @C=.7em @R=.7em @!R {
& \ctrl{1} & \qw &\targ & \qw & \targ & \qw & \ctrl{1} & \qw \\
& \targ & \targ & \qw & \qw & \qw & \targ & \targ & \qw \\
& \targ & \ctrl{-1} & \qw & \gate{H} & \qw & \ctrl{-1} & \targ & \qw \\
& \ctrl{-1} & \qw & \ctrl{-3} & \qw & \ctrl{-3} & \qw & \ctrl{-1}& \qw & \push{\rule{0em}{1.46em}} 
}
\raisebox{-3.6em}{\rule{.8em}{0em}$=$\rule{1.1em}{0em}} 
\Qcircuit @C=.7em @R=.7em @!R {
& \qw & \qw & \qw & \qw & \qw & \qw \\
& \qw & \targ & \qw & \targ & \qw & \qw \\
& \targ & \ctrl{-1} & \gate{H} & \ctrl{-1} & \targ & \qw \\
& \ctrl{-1} & \qw & \qw & \qw & \ctrl{-1}& \qw & \push{\rule{0em}{1.46em}} 
}
\raisebox{-3.6em}{\rule{.8em}{0em}$=$\rule{1.1em}{0em}} 
\Qcircuit @C=.7em @R=.7em @!R {
& \qw & \qw & \qw & \qw & \qw & \qw \\
& \gate{H} & \qw & \ctrl{1} & \ctrl{1} & \gate{H} & \qw \\
& \gate{H} & \ctrl{1} & \targ & \ctrl{-1} & \targ & \qw \\
& \qw & \ctrl{-1} & \qw & \qw & \ctrl{-1}& \qw & \push{\rule{0em}{1.46em}} 
}
\raisebox{-3.6em}{\rule{.8em}{0em}$=$\rule{1.1em}{0em}} 
\Qcircuit @C=.7em @R=.7em @!R {
& \qw & \qw & \qw & \qw & \qw & \qw & \qw \\
& \gate{H} & \gate{S} & \ctrl{2} & \ctrl{1} & \qw & \gate{H} & \qw \\
& \gate{H}& \qw & \qw & \gate{Y} & \gate{Y} & \qw & \qw  \\
& \qw & \gate{S^{\dagger}} & \ctrl{-2} & \qw & \ctrl{-1}& \qw & \qw & \push{\rule{0em}{1.46em}} 
}
\\ \\
b) &
\hspace{-1.68em}
\Qcircuit @C=.7em @R=.7em @!R {
\push{\rule{2.7em}{0em}}& \lstick{\ket{+}} & \ctrl{3} & \qw & \qw & \qw & \qw & \qw & \qw & \ctrl{3} & \measureD{X} \\
& \qw & \qw & \qw & \qw & \qw & \qw & \qw & \qw & \qw  & \qw & \qw \\
& \qw & \qw & \gate{H} & \gate{S} & \ctrl{2} & \ctrl{1} & \qw & \gate{H} & \qw & \qw & \qw \\
& \qw & \gate{H} & \gate{H}& \qw & \qw & \gate{Y} & \gate{Y} & \qw & \gate{H} &  \qw & \qw \\
& \qw & \qw & \qw & \gate{S^{\dagger}} & \ctrl{-2} & \qw & \ctrl{-1}& \qw & \qw & \qw & \qw & \push{\rule{0em}{1.46em}}
}
\raisebox{-4.6em}{\rule{1.6em}{0em}$=$\rule{1.9em}{0em}} 
\Qcircuit @C=.7em @R=.7em @!R {
& & \lstick{\ket{+}} & \qw & \ctrl{3} & \qw & \qw & \ctrl{3} & \qw & \measureD{X} \\
& \qw & \qw & \qw & \qw & \qw & \qw & \qw & \qw  & \qw & \qw \\
& \gate{H} & \gate{S} & \ctrl{2} & \qw & \ctrl{1} & \qw & \qw & \gate{H} & \qw & \qw \\
& \gate{H}& \qw & \qw & \gate{H} & \gate{Y} & \gate{Y} & \gate{H} & \qw & \qw  & \qw \\
& \qw & \gate{S^{\dagger}} & \ctrl{-2} & \qw & \qw & \ctrl{-1} & \qw & \qw & \qw & \qw & \push{\rule{0em}{1.46em}}
}
\\ \\
& 
\raisebox{-4.6em}{\rule{8em}{0em}$=$\rule{1.9em}{0em}} 
\Qcircuit @C=.7em @R=.7em @!R {
& & \lstick{\ket{+}} &\qw & \ctrl{3} & \qw & \ctrl{3} & \qw & \ctrl{3} & \qw & \measureD{X} & \\
& \qw & \qw & \qw & \qw & \qw & \qw & \qw & \qw & \qw  & \qw & \qw \\
& \gate{H} & \gate{S} & \ctrl{2} & \qw & \ctrl{1} & \qw & \qw & \qw & \gate{H} & \qw  & \qw \\
& \gate{H}& \qw & \qw & \control \qw & \gate{Y} & \gate{H} & \gate{Y} & \gate{H} & \qw  & \qw  & \qw \\
& \qw & \gate{S^{\dagger}} & \ctrl{-2} & \qw & \qw & \qw & \ctrl{-1} & \qw & \qw & \qw & \qw  & \push{\rule{0em}{1.46em}}
}
\raisebox{-4.6em}{\rule{1.6em}{0em}$=$\rule{1.9em}{0em}} 
\Qcircuit @C=.7em @R=.7em @!R {
& & \lstick{\ket{+}} & \qw & \ctrl{2} & \qw & \qw & \ctrl{4} & \qw & \measureD{X} \\
& \qw & \qw & \qw & \qw & \qw & \qw & \qw & \qw & \qw & \qw \\
& \gate{H} & \gate{S} & \ctrl{2} \qw & \ctrl{-2} & \ctrl{1} & \qw & \qw & \gate{H} & \qw & \qw \\
& \gate{H}& \qw & \qw & \qw & \gate{Y} & \gate{Y} & \qw & \qw  & \qw & \qw \\
& \qw & \gate{S^{\dagger}} & \ctrl{-2} & \qw & \qw & \ctrl{-1} & \ctrl{-4} & \qw & \qw & \qw & \push{\rule{0em}{1.46em}}
}
\end{tabular}
\caption{Additional details regarding the relations between circuits
  in Fig.~\ref{fig:logicalMeas}.  (a) A sequence of circuit identities
  deriving the form of the encoded Hadamard gate used in
  Fig.~\ref{fig:logicalMeas}(b). The starting circuit implements $H$ on
  the second logical qubit of the code $\code{C}_4$ by decoding the
  logical qubits into the first and third physical qubits, applying
  $H$ to the third qubit, and re-encoding. The first equivalence is
  obtained by commuting and cancelling pairs of $\protect\CX$ gates.
  The second equivalence uses the decomposition of the $\protect\CX$
  gate into $\protect\CZ$ and $H$ gates several times as well as the
  identity $H^2=I$.  The final equivalence uses two facts: $ZX=i Y$ and the order
  of a $\protect\CX$ and a $\protect\CZ$ gate with the same target can
  be exchanged if a $\protect\CZ$ gate is added between the controls.  
 (b) A sequence of circuit identities showing why the
  pair of $\protect\CH$ gates targeting the fourth qubit in
  Fig.~\ref{fig:logicalMeas}(b) can be eliminated.  Other than reorganization
  of commuting gates, these equivalences rely on the fact that, because $H$ anticommutes with $Y$,
  $\protect \CH$ and $\protect \CY$ gates with the same target can be
  exchanged if a $\protect\CZ$ gate between the two controls is added.
  \label{fig:CircuitDetails}}
\end{figure*}
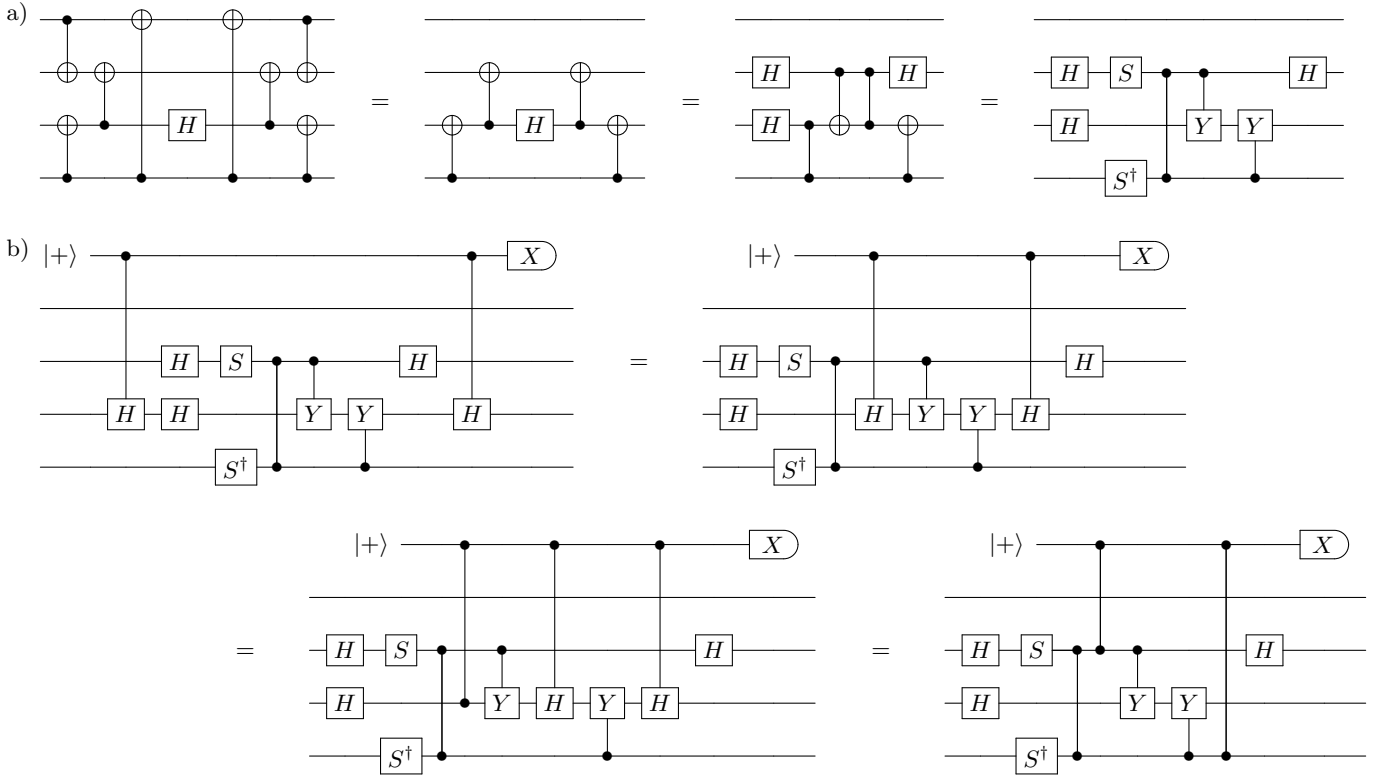

Fig.~\ref{fig:CircuitDetails} provides additional details about the circuit identities used in Fig.~\ref{fig:logicalMeas}.

\end{document}